%% file: 2015-ARXIV-AdaptiveComputationOfTheSwapInsertEditionDistance-BarbayPerez.tex
\documentclass{llncs}

\usepackage{versions}
\excludeversion{INUTILE}
\includeversion{LONG}\excludeversion{SHORT}
\excludeversion{TODO}\excludeversion{JEREMY}\excludeversion{PABLO}

\usepackage[english]{babel}
\usepackage{amsmath}
\usepackage{amssymb}
\newtheorem{observation}{Observation}
\usepackage{graphicx}
\usepackage{subfig}
\usepackage{algo}

\providecommand{\alphabetSize}{d}
\newcommand{\D}{d}

\providecommand{\dist}{\delta}
\providecommand{\DIST}{DIST}
\providecommand{\COUNTERS}{\mathbb{W}}

\providecommand{\pb}[1]{{\sc #1} problem}
\providecommand{\implies}{~\Longrightarrow~}

\title{Adaptive Computation  of the \\ Swap-Insert Correction Distance}

\author{
	J\'er\'emy Barbay\inst{1}\thanks{
		Partially supported by Millennium Nucleus Information and Coordination in Networks ICM/FIC RC130003. Part of this work was presented at the conference SPIRE 2015~\cite{2015-SPIRE-AdaptiveComputationOfTheSwapInsertEdutionDistance-BarbayPerez}.
	} 
    \and
    Pablo P\'erez-Lantero\inst{2}$^{\star}$
}

\institute{
	Departamento de Ciencias de la Computaci\'on, Universidad de Chile, Chile
	\and
	Escuela de Ingenier\'ia Civil en Inform\'atica, Universidad de Valpara\'iso, Chile
}

\begin{document}
\maketitle
\begin{abstract}
The Swap-Insert Correction distance from a string $S$ of length $n$ to another string $L$ of length $m\geq n$ on the alphabet $[1..d]$ is the minimum number of insertions, and swaps of pairs of adjacent symbols, converting $S$ into $L$. Contrarily to other correction distances, computing it is NP-Hard in the size $d$ of the alphabet. We describe an algorithm computing this distance in time within $O(d^2 nm g^{d-1})$, where there are $n_\alpha$ occurrences of $\alpha$ in $S$, $m_\alpha$ occurrences of $\alpha$ in $L$, and where $g=\max_{\alpha\in[1..d]} \min\{n_\alpha,m_\alpha-n_\alpha\}$ measures the difficulty of the instance. The difficulty $g$ is bounded by above by various terms, such as the length of the shortest string $S$, and by the maximum number of occurrences of a single character in $S$.  Those results illustrate how, in many cases, the correction distance between two strings can be easier to compute than in the worst case scenario.
\end{abstract}

\keywords{
Adaptive, 
Dynamic Programming,
Edit Distance,
Insert,
Swap.
}

\begin{INUTILE}
ABSTRACT for Submission:

The Swap-Insert String-to-String Correction distance from a string $S$ to another string $L$ on the alphabet $[1..d]$ is the minimum number of insertions and swaps of pairs of adjacent symbols converting $S$ into $L$.  We describe an algorithm computing this distance in time polynomial in the lengths $n$ of $S$ and $m$ of $L$, so that it is as good as previous results, and which additionally takes advantage of ``easy'' instances such that for each symbol $\alpha\in[1..d]$, either $m_\alpha-n_\alpha$ is small or $n_\alpha$ is small, where there are $n_\alpha$ occurrences of $\alpha$ in $S$ and $m_\alpha$ occurrences of $\alpha$ in $L$.

\keywords{
Adaptive
Dynamic Programming
Edit Distance
Insert
Swap
}

\end{INUTILE}

\input{editDistance}

\section*{\ackname}
The authors would like to thank the anonymous referees of SPIRE 2015 for insightful comments.

\bibliography{biblio}

\end{document}

%% file: editDistance.tex
\section{Introduction}\label{sec:intro}

Given two strings $S$ and $L$ on the alphabet $\Sigma=[1..d]$ and a list
of correction operations on strings, the {\sc String-to-String Correction}
distance is the minimum number of operations required to transform the
string~$S$ into the string~$L$. Introduced in 1974 by Wagner and
Fischer~\cite{wagner1974string}, this concept has many applications, from
suggesting corrections for typing mistakes, to decomposing the changes
between two consecutive versions into a minimum number of correction
steps, for example within a control version system\begin{LONG} such as
{\tt cvs}, {\tt svn} or {\tt git}\end{LONG}.

Each distinct set of correction operators yields a distinct correction
distance on strings.  For instance, Wagner and
Fischer~\cite{wagner1974string} showed that for the three following
operations, the \texttt{insertion} of a symbol at some arbitrary position,
the \texttt{deletion} of a symbol at some arbitrary position, and the
\texttt{substitution} of a symbol at some arbitrary position, there is a
dynamic program solving this problem in time within $O(nm)$ when $S$ is of
length $n$ and $L$ of length $m$.  Similar complexity results, all
polynomial, hold for many other different subsets of the natural
correction operators, with one striking exception:
Wagner~\cite{wagner1975complexity} proved the NP-hardness of the {\sc
  Swap-Insert Correction} distance, denoted $\delta(S,L)$ through this
paper, i.e. the correction distance when restricted to the operators {\tt
  insertion} and {\tt swap} (or, by symmetry, to the operators {\tt
  deletion} and {\tt swap}). 

The {\sc Swap-Insert Correction} distance's difficulty attracted special
interest, with two results of importance: Abu-Khzam et
al.~\cite{AbuKhzam201141} described an algorithm computing $\dist(S,L)$ in
time within $O({1.6181}^{\dist(S,L)} m)$, and
Meister~\cite{2015-TCS-UsingSwapsAndDeletesToMakeStringsMatch-Meister}
described an algorithm computing $\dist(S,L)$ in time polynomial in the
input size when $S$ and $L$ are strings on a finite alphabet\begin{LONG}:
its running time is $(m+1)^{2d+1}\cdot (n+1)^2$ times some polynomial
function on $n$ and $m$\end{LONG}.

The complexity of Meister's
result~\cite{2015-TCS-UsingSwapsAndDeletesToMakeStringsMatch-Meister},
polynomial in $m$ of degree $2d+1$, is a very pessimistic approximation of
the computational complexity of the distance. At one extreme, the
\textsc{Swap-Insert Correction} distance between two strings which are
very similar (e.g. only a finite number of symbols need to be swapped or
inserted) can be computed in time linear in $n$ and $d$. At the other
extreme, the \textsc{Swap-Insert Correction} distance of strings which are
completely different (e.g. their effective alphabets are disjoint) can
also be computed in linear time (it is then close to $n+m$). Even when $S$
and $L$ are quite different, $\delta(S,L)$ can be ``easy'' to compute:
when mostly swaps are involved to transform $S$ into $L$ (i.e. $S$ and $L$
are almost of the same length), and when mostly insertions are involved to
transform $S$ into $L$ (i.e. many symbols present in $L$ are absent from
$S$).

\noindent\textbf{Hypothesis:}
\label{sec:hypothesis}
We consider whether the \textsc{Swap-Insert Correction} distance
$\dist(S,L)$ can be computed in time polynomial in the length of the input
strings for a constant alphabet size, while still taking advantage of
cases such as those described above, where the distance $\dist(S,L)$ can
be computed much faster.

\noindent\textbf{Our Results:}
\label{sec:our-results}
After a short review of previous results and techniques in
Section~\ref{sec:previous-work}, we present such an algorithm in
Section~\ref{sec:algorithm}, in four steps: the intuition behind the
algorithm in Section~\ref{sec:highlevel},
the formal description of
the dynamic program in Section~\ref{sec:recursive}, and the formal
analysis of its complexity in Section~\ref{sec:dp}.
\def\K{\sum_{\alpha=1}^{\D}(m_{\alpha}-g_{\alpha})}
In the latter, we define the local imbalance
$g_{\alpha}=\min\{n_{\alpha},m_{\alpha}-n_{\alpha}\}$ for each symbol
$\alpha\in\Sigma$, summarized by the global imbalance measure
$g=\max_{\alpha\in\Sigma} g_\alpha$, and prove that our algorithm runs in
time within
\begin{LONG}
\[
 O\left(
  d(n+m) 
  +d^2n
  \cdot  \K\cdot
    \prod_{\alpha\in\Sigma_{+}}(g_\alpha+1)
  \right),
\]
in the worst case over all instances of fixed sizes $n$ and $m$, with
imbalance vector $(g_1,\ldots,g_d)$; where
$\Sigma_{+}=\{\alpha\in\Sigma: g_{\alpha}>0\}$ if $g_\alpha=0$ for any
$\alpha\in\Sigma$, and
$\Sigma_{+}=\Sigma\setminus \{\arg\min_{\alpha\in\Sigma}g_{\alpha}\}$
otherwise.
This simplifies to within
\end{LONG}
 $O(d^2 g^{d-1} n m )$ in the worst case over
instances where $d,n,m$ and $g$ are fixed.

We discuss in Section~\ref{sec:conclusions} some implied results, and some
questions left open\begin{LONG}, such as when the operators are assigned
asymmetric costs, when the algorithm is required to output the sequence of
corrections, when only swaps are allowed, or when the distribution of the
frequencies of the symbol is very unbalanced\end{LONG}.  
\begin{SHORT}
Additional details are deferred to the full
version~\cite{2015-ARXIV-AdaptiveComputationOfTheSwapInsertEdutionDistance-BarbayPerez}.
\end{SHORT}

\section{Background}
\label{sec:previous-work}

In 1974, motivated by the problem of correcting typing and transmission
errors, Wagner and Fischer~\cite{wagner1974string} introduced the
\pb{String-to-String Correction}, which is to compute the minimum number
of corrections required to change the source string $S$ into the target
string $L$. They considered the following operators:
the \texttt{insertion} of a symbol at some arbitrary position,
the \texttt{deletion} of a symbol at some arbitrary position, and
the \texttt{substitution} of a symbol at some arbitrary position.
They described a dynamic program solving this problem in time within
$O(nm)$ when $S$ is of length $n$ and $L$ of length $m$.  The worst case
among instances of fixed input size $n+m$ is when $n=m/2$, which yields a
complexity within $O(n^2)$.

\begin{TODO}
Mention some adaptive results on other distances.
\end{TODO}

In 1975, 
Lowrance and
Wagner~\cite{wagner1975extension} extended the \textsc{String-to-String
  Correction} distance to the cases where one considers not only the
\texttt{insertion}, \texttt{deletion}, and \texttt{substitution}
operators, but also the \texttt{swap} operator, which exchanges the
positions of two contiguous symbols.  Not counting the identity, fifteen
different variants arise when considering any given subset of those four
correction operators. Thirteen of those variants can be computed in
polynomial
time~\cite{wagner1975complexity,wagner1974string,wagner1975extension}.
The two remaining distances, the computation of the \textsc{Swap-Insert
  Correction} distance and its symmetric the \textsc{Swap-Delete
  Correction} distance, are equivalent by symmetry, and are NP-hard to
compute~\cite{wagner1975complexity}, hence our interest. All our results
on the computation of the \textsc{Swap-Insert Correction} distance from
$S$ to $L$ directly imply the same results on the computation of the
\textsc{Swap-Delete Correction} distance from $L$ to $S$.

\begin{LONG}
In 2011, Abu-Khzam et al.~\cite{AbuKhzam201141} described an algorithm
computing the \textsc{Swap-Delete Correction} distance from a string $L$
to a string $S$ (and hence the \textsc{Swap-Insert Correction} from $S$ to
$L$). Their algorithm decides if this distance is at most a given
parameter $k$, in time within $O({1.6181}^k m)$. This indirectly yields an
algorithm computing both distances in time within
$O({1.6181}^{\delta(S,L)} m)$: testing values of $k$ from $0$ to infinity
in increasing order yields an algorithm computing the distance in time
within
$O(\sum_{k=0}^{\delta(S,L)}{1.6181}^k m) \subset
O({1.6181}^{\delta(S,L)}m)$.
Since any correct algorithm must verify the correctness of its output,
such an algorithm implies the existence of an algorithm with the same
running time which outputs a minimum sequence of corrections from $S$ to
$L$.  Later in 2013, Watt~\cite{watt2013} showed that computing the
\textsc{Swap-Deletion Correction} distance has a kernel size of $O(k^4)$.
\end{LONG}

In 2013, Spreen~\cite{spreen2013} observed that Wagner's NP-hardness
proof~\cite{wagner1975complexity} was based on unbounded alphabet
sizes (i.e.\ the \pb{Swap-Insert Correction} is NP-hard when the size
$d$ of the alphabet is part of the input), and suggested that this
problem might be tractable for fixed alphabet sizes. He described some
polynomial-time algorithms for various special cases when the alphabet is binary, and 
some more general properties.

In 2014, Meister~\cite{2015-TCS-UsingSwapsAndDeletesToMakeStringsMatch-Meister} extended Spreen's
work~\cite{spreen2013} to an algorithm computing the
\textsc{Swap-Insert Correction} distance from a string $S$ of length
$n$ to another string $L$ of length $m$ on any fixed alphabet size
$\D\ge 2$, in time polynomial in $n$ and $m$.  The algorithm is explicitly
based on finding an injective function
$\varphi:[1.. n]\rightarrow [1.. m]$ such that $\varphi(i)=j$ if and only
if $S[i]=L[j]$, and the total number of crossings is minimized. Two
positions $i<i'$ of $S$ define a {\em crossing} if and only if
$\varphi(i)>\varphi(i')$.  Such a number of crossings equals the number of
swaps, and the number of insertions is always
equal to $m-n$. Meister proved that the time complexity of this algorithm
is equal to $(m+1)^{2d+1}\cdot (n+1)^2$ times some function polynomial in
$n$ and $m$.

\begin{LONG}
We describe in the following section an algorithm computing the {\sc
  Swap-Insert Correction} distance in explicit polynomial
time, and which running time goes gradually down to linear for easier
cases.
\end{LONG}

\section{Algorithm}\label{sec:algorithm}

We describe the intuition behind our algorithm in
Section~\ref{sec:highlevel}, 
the high level description of the
dynamic program in Section~\ref{sec:recursive}, 
\begin{LONG}
the full code of the algorithm in Section~\ref{sec:complete-algorithm}
\end{LONG}
and the formal analysis of its complexity in Section~\ref{sec:dp}.

\subsection{High level description}
\label{sec:highlevel}

The algorithm runs through $S$ and $L$ from left to right, building a
mapping from the characters of $S$ to a subset of the characters of $L$,
using the fact that, for each distinct character, the mapping function on
positions is monotone.  The Dynamic Programming matrix has size
$n_1\times\cdots\times n_d<n^d$.

For every string $X\in\{S,L\}$, let $X[i]$ denote the $i$-th symbol of $X$
from left to right ($i\in[1..|X|]$), and $X[i..j]$ denote the
substring of $X$ from the $i$-th symbol to the $j$-th symbol ($1\le i\le j \le |X|$). 
For every $1\leq j<i\leq n$, let $X[i..j]$ denote the empty string.
Given any symbol $\alpha\in\Sigma$, let $rank(X,i,\alpha)$ denote the
number of occurrences of the symbol $\alpha$ in the substring $X[1..i]$, and
$select(X,k,\alpha)$ denote the value $j\in[1..|X|]$ such that the $k$-th
occurrence of $\alpha$ in $X$ is precisely at position $j$, if $j$
exists. If $j$ does not exist, then $select(X,k,\alpha)$ is $null$.

The algorithm runs through $S$ and $L$ simultaneously from left to right,
skipping positions where the current symbol of $S$ equals the current
symbol of $L$, and otherwise branching out between two options to correct
the current symbol of $S$: inserting a symbol equal to the current symbol
of $L$ in the current position of $S$, or moving (by applying many swaps)
the first symbol of the part not scanned of $S$ equal to the current
symbol of $L$, to the current position in $S$. 

More formally, the computation of $\delta(S,L)$ can be reduced to the
application of four rules:
\begin{itemize}
\item {\bf if $S$ is empty}: We just return the length $|L|$ of $L$, since
  insertions are the only possible operations to perform in $S$.

\item {\bf if some $\alpha\in\Sigma$ appears more times in $S$ than in $L$}: We return $+\infty$, since
{\tt delete} operations are not allowed to make $S$ and $L$ match.

\item {\bf if $S$ and $L$ are not empty, $S[1]=L[1]$}: We return $\delta(S[2..|S|],L[2..|L|])$.

\item {\bf if $S$ and $L$ are not empty, $S[1]\neq L[1]$}: We compute two distances:
the distance $d_{ins}=1+\delta(S,L[2..|L|])$ corresponding to an {\tt insertion} of the
symbol $L[1]$ at the first position of $S$, and the distance $d_{swaps}=(r-1)+\delta(S',L[2..|L|])$
corresponding to perform $r-1$ {\tt swaps} to bring to the first position of $S$ the first symbol
of $S$ equal to $L[1]$. In this case, $r$ denotes the position of such a symbol, and $S'$ the
string resulting from $S$ by removing that symbol. We then return $\min\{d_{ins},d_{swaps}\}$.

\end{itemize}

There can be several overlapping subproblems in the recursive definition
of $\delta(S,L)$ described above, which calls for {\em dynamic
  programming}~\cite{Cormen2009} and {\em memoization}\begin{LONG}\footnote{Cormen et al. \cite{Cormen2009} explain that {\em memoization}
  comes from \emph{memo}, referring to the fact that the technique
  consists in recording a value so that we can look it up later.}\end{LONG}.
In any call $\delta(S',L')$ in the recursive computation of
$\delta(S,L)$, the string $L'$ is always a substring $L[j..|J|]$ for
some $j\in[1..|J|]$, and can thus be replaced by such an index $j$,
but this is not always the case for the string $S'$. Observe that $S'$
is a substring $S[i..|S|]$ for some $i\in[1..|S|]$ with (eventually)
some symbols removed.  Furthermore, if for some symbol $\alpha\in\Sigma$
precisely $c_{\alpha}$ symbols $\alpha$ of $S[i..|S|]$ have been
removed, then those symbols are precisely the first $c_{\alpha}$
symbols $\alpha$ from left to right.  We can then represent $S'$ by
the index $i$ and a counter $c_{\alpha}$ for each symbol
$\alpha\in\Sigma$ of how many symbols $\alpha$ of $S[i..|S|]$ are
removed (i.e.\ ignored).  In the above fourth rule, the position $r$
is equivalent to the position of the $(c_{L[1]}+1)$-th occurrence of
the symbol $L[1]$ in $S[i..|S|]$.  To quickly compute $r$, the
functions $rank$ and $select$ will be used.




Let $\COUNTERS=\prod_{\alpha=1}^{\D}[0..n_{\alpha}]$ denote the domain of
such vectors of counters, where for any
$\overline{c}=(c_1,c_2,\ldots,c_{\D})\in\COUNTERS$, $c_{\alpha}$ denotes
the counter for $\alpha\in\Sigma$.  Using the ideas described above, the
algorithm recursively computes the extension $\DIST(i,j,\overline{c})$ of
$\delta(S,L)$, defined for each $i\in[1..n+1]$, $j\in[1..m+1]$, and
$\overline{c}=(c_1,c_2,\ldots,c_{\D})\in\COUNTERS$, as the value of
$\dist(S[i..n]_{\overline{c}},L[j..m])$, where $S[i..n]_{\overline{c}}$ is
the string obtained from $S[i..n]$ by removing (i.e.\ ignoring) for each
$\alpha\in\Sigma$ the first $c_{\alpha}$ occurrences of $\alpha$ from left
to right.

Given this definition, \( \dist(S,L) = \DIST(1,1,\overline{0}), \)
where $\overline{0}$ denotes the vector $(0,\ldots,0)\in\COUNTERS$.  Given
$i$, $j$, and $\overline{c}$, 
$\DIST(i,j,\overline{c})<+\infty$ if and only if for each symbol
$\alpha\in\Sigma$ the number of considered (i.e.\ not removed or ignored)
$\alpha$ symbols in $S[i..n]$ is at most the number of $\alpha$ symbols in
$L[j..m]$. That is,
\( count(S,i,\alpha)-c_{\alpha}~\le~ count(L,j,\alpha) \)
for all $\alpha\in\Sigma$, where
$count(X,i,\alpha)=rank(X,|X|,\alpha)-rank(|X|,i-1,\alpha)$ is the number
of symbols $\alpha$ in the string $X[i..|X|]$.  In the following, we show
how to compute $\DIST(i,j,\overline{c})$ recursively for every $i$, $j$,
and $\overline{c}$. For a given $\alpha\in\Sigma$, let
$\overline{w}_{\alpha}\in\COUNTERS$ be the vector whose components are all
equal to zero except the $\alpha$-th component that is equal to 1.

\subsection{Recursive computation of $\DIST(i,j,\overline{c})$}
\label{sec:recursive}

We will use the following observation which considers the \texttt{swap} operations
performed in the optimal transformation from a
short string $S$ of length $n$ to a larger string $L$ of length $m$.
\begin{observation}[\cite{AbuKhzam201141,spreen2013}]
\label{obs:swaps}
The \texttt{swap} operations used in any optimal solution satisfy the
following properties: two equal symbols cannot be swapped; each symbol is
always swapped in the same direction in the string; and if some symbol is
moved from some position to another by performing \texttt{swaps}
operations, then no symbol equal to it can be inserted afterwards between
these two positions.
\end{observation} 
The following lemma deals with the basic case where $S[i..n]$ and
$L[j..m]$ start with the same symbol, i.e. $S[i]=L[j]$. When the beginnings
of both strings are the same, matching those two symbols seems like an
obvious choice in order to minimize the distance, but one must be careful
to check first if the first symbol from $S[i..n]$ has not been scheduled to be
``swapped'' to an earlier position, in which case it must be ignored and
skipped:

\begin{lemma}\label{lem:core1}
Given two strings $S$ and $L$ over the alphabet $\Sigma$,
for any positions $i\in[1..n]$ in $S$ and $j\in[1..m]$ in
$L$, for any vector of counters
$\overline{c}=(c_1,\ldots,c_{\D})\in\COUNTERS$ and for any symbol
$\alpha\in\Sigma$,
\[
\left.
  \begin{array}{l}
  S[i]=L[j]=\alpha \\
  c_{\alpha}=0
  \end{array}
\right\} 
~\Longrightarrow~
\DIST(i,j,\overline{c})~=~\DIST(i+1,j+1,\overline{c}).
\]
\end{lemma}

\begin{proof}
Given strings $X,Y$ in the alphabet $\Sigma$, and an integer $k$,
Abu-Khzam et al.~\cite[Corollary 1]{AbuKhzam201141} proved 
that if $X[1]=Y[1]$, then:
\[
	\delta(X,Y) ~\le~ k~\text{if and only if}~\delta(X[2..|X|],Y[2..|Y|]) ~\le~ k.
\]
Given that one option to transform $X$ into $Y$
with the minimum number of operations is to transform $X[2..|X|]$ into
$Y[2..|Y|]$ with the minimum number of operations (matching $X[1]$ with
$Y[1]$), we have:
\[
	\delta(X,Y) ~\le~ \delta(X[2..|X|],Y[2..|Y|]).
\]
By selecting $k=\delta(X,Y)$, we obtain the equality
\[
	\delta(X,Y) ~=~ \delta(X[2..|X|],Y[2..|Y|]).
\]
Then, since the symbol $\alpha=S[i]$ must be considered (because $c_{\alpha}=0$),
and $S[i]=L[j]$, we can apply the above statement for
$X=S[i..n]_{\overline{c}}$ and $Y=L[j..m]$ to obtain the next equalities:
\[
	\DIST(i,j,\overline{c}) = \delta(X,Y) ~=~ \delta(X[2..|X|],Y[2..|Y|]) = \DIST(i+1,j+1,\overline{c}).
\]
The result thus follows.\qed
\end{proof}

The second simplest case is when the first available symbol of $S[i..n]$
is already matched (through {\tt swaps}) to a symbol from $L[1..j-1]$. The
following lemma shows how to simply skip such a symbol:

\begin{lemma}\label{lem:core2}
Given $S$ and $L$ over the alphabet $\Sigma$,
for any positions $i\in[1..n]$ in $S$ and $j\in[1..m]$ in
$L$, and for any vector of counters
$\overline{c}=(c_1,\ldots,c_{\D})\in\COUNTERS$ and for any symbol
$\alpha\in\Sigma$,
\[
\left.
  \begin{array}{l}
  S[i]=\alpha \\
  c_{\alpha}>0
  \end{array}
\right\} 
 \implies 
	\DIST(i,j,\overline{c}) = \DIST(i+1,j,\overline{c}-\overline{w}_{\alpha}).
\]
\end{lemma}

\begin{proof}
Since $c_{\alpha}>0$, the first $c_{\alpha}$ symbols $\alpha$ of $S[i..n]$
have been ignored, thus $S[i]$ is ignored. Then, $\DIST(i,j,\overline{c})$
must be equal to $\DIST(i+1,j,\overline{c}-\overline{w}_{\alpha})$, case
in which $c_{\alpha}-1$ symbols $\alpha$ of $S[i+1..n]$ are ignored.
\qed\end{proof}

The most important case is when the first symbols of $S[i..n]$ and
$L[j..m]$ do not match: the minimum ``path'' from $S$ to $L$ can then
start either by an \texttt{insertion} or a \texttt{swap} operation.

\begin{lemma}\label{lem:core3}
Given $S$ and $L$ over the alphabet $\Sigma$, for any
positions $i\in[1..n]$ in $S$ and $j\in[1..m]$ in $L$, and for any vector
of counters $\overline{c}=(c_1,\ldots,c_{\D})\in\COUNTERS$,
note $\alpha,\beta\in\Sigma$ the symbols $\alpha=S[i]$ and
$\beta=L[j]$,
$r$ the position $r=select(S,rank(S,i,\beta)+c_{\beta}+1,\beta)$ in $S$ of
the $(c_{\beta}+1)$-th symbol $\beta$ of $S[i..n]$, and 
$\Delta$ the number $\sum_{\theta=1}^{\D}\min\{c_{\theta},$
$rank(S,r,\theta)-rank(S,i-1,\theta)\}$ of symbols ignored in $S[i..r]$.

If $\alpha\neq\beta$ and $c_{\alpha}=0$, then
$\DIST(i,j,\overline{c})=\min\{d_{ins},d_{swaps}\}$, where
\[
d_{ins} ~=~ \left\{\begin{array}{ll}
  \DIST(i,j+1,\overline{c}) + 1 & \text{ if } c_{\beta}=0\\
  +\infty & \text{ if } c_{\beta}>0\\
	\end{array}\right.
\]
and
\[
d_{swaps} ~=~ \left\{\begin{array}{ll}
  (r-i)-\Delta+\DIST(i,j+1,\overline{c}+\overline{w}_{\beta}) & \text{ if } r\neq 0\\
  +\infty & \text{ if } r= 0.\\
	\end{array}\right.
\]
\end{lemma}

\begin{proof}
Let $S'[1..n']=S[i..n]_{\overline{c}}$.  
Given that $\alpha\neq\beta$ and $c_{\alpha}=0$, there are two
possibilities for $\DIST(i,j,\overline{c})$: (1) transform $S'[1..n']$
into $L[j+1..m]$ with the minimum number of operations, and after that
insert a symbol $\beta$ at the first position of the resulting
$S'[1..n']$; or (2) swap the first symbol $\beta$ in $S'[2..n']$ from left
to right from its current position $r'$ to the position $1$ performing
$r'-1$ \texttt{swaps}, and then transform the resulting $S'[2..n']$ into
$L[j+1..m]$ with the minimum number of operations.
Observe that option (1) can be performed if and only if there
is no symbol $\beta$ ignored in $S[i..n]$ (see Observation~\ref{obs:swaps}). If this is the case, then
$\DIST(i,j,\overline{c})=\DIST(i,j+1,\overline{c}) + 1$.  Option (2) can
be used if and only if there is a non-ignored symbol $\beta$ in $S[i..n]$,
where the first one from left to right is precisely at position
$r=select(S,rank(S,i,\beta)+c_{\beta}+1,\beta)$. In such a case
$r'=(r-i+1)-\Delta$, where $\Delta=\sum_{\theta=1}^{\D}\min\{c_{\theta},$
$rank(S,r,\theta)-rank(S,i-1,\theta)\}$ is the total number of ignored
symbols in the string $S[i..r]$.  Hence, the number of swaps\texttt{}
counts to $r'-1=(r-i)-\Delta$.  Then, the correctness of $d_{ins}$,
$d_{swaps}$, and the result follow.  \qed\end{proof}

The next two lemmas deal with the cases where one string is completely
processed.  When $L$ has been completely processed, either the remaining
symbols in $S$ have all previously been matched via \texttt{swaps} and the
distance equals zero, or there is no sequence of operations correcting $S$
into $L$:

\begin{lemma}\label{lem:easy1}
Given $S$ and $L$ over the alphabet $\Sigma$, for any
positions $i\in[1..n+1]$ in $S$ and $j\in[1..m]$ in $L$, for any vector of
counters $\overline{c}=(c_1,\ldots,c_{\D})\in\COUNTERS$,
\[
\DIST(i,m+1,\overline{c})= \left\{ 
  \begin{array}{cl}
  0 & \mbox{ if $c_1+\ldots+c_{\D}=n-i+1$ and} \\
  +\infty &  \mbox{ otherwise.}
  \end{array}\right.
\]
\end{lemma}

\begin{proof}
Note that $\DIST(i,m+1,\overline{c})$ is the minimum number of operations
to transform the string $S[i..n]$ into the empty string $L[m+1..m]$. This
number is null if and only if all the $n-i+1$ symbols of $S[i..n]$ have
been ignored, that is, $c_1+\ldots+c_{\D}=n-i+1$. If not all the symbols
have been ignored, then such a transformation does not exist and
$\DIST(i,m+1,\overline{c})=+\infty$.  \qed\end{proof}

When $S$ has been completely processed, there are only insertions left to
perform: the distance can be computed in constant time, and the list of
corrections in linear time.

\begin{lemma}\label{lem:easy2}
Given $S$ and $L$ over the alphabet $\Sigma$,
for any position $j\in[1..m+1]$ in
$L$, and for any vector of counters
$\overline{c}=(c_1,\ldots,c_{\D})\in\COUNTERS$,
\[
\DIST(n+1,j,\overline{c})= \left\{ 
  \begin{array}{cl}
  m-j+1 & \mbox{ if
$\overline{c}=\overline{0}$  and} \\
  +\infty &  \mbox{ otherwise.}
  \end{array}\right.
\]
\end{lemma}

\begin{proof}
Note that $\DIST(i,m+1,\overline{c})$ is the minimum number of operations
to transform the empty string $S[n+1..n]$ into the string $L[j..m]$. If
$\overline{c}=\overline{0}$, then $\DIST(n+1,j,\overline{c})<+\infty$ and
the transformation consists of only insertions which are $m-j+1$.  If
$\overline{c}\neq\overline{0}$, then $\DIST(n+1,j,\overline{c})=+\infty$.
\qed\end{proof}

\begin{figure}[t]
	\begin{algorithm}{}{$\DIST(i,j,\overline{c}=(c_1,\ldots,c_{\D}))$}
		\qif $\DIST(i,j,\overline{c})=+\infty$ \qthen\label{L1}\\
			\qreturn $+\infty$
		\qelif $i=n+1$ \qthen\\
				(* insertions *)\\
				\qreturn $m-j+1$\label{L2}
		\qelif $j=m+1$ \qthen\label{L3}\\
				(* skip all symbols since they were ignored *)\\
				\qreturn $0$\label{L4}
		\qelse\\
			$\alpha\leftarrow S[i]$, $\beta\leftarrow L[j]$\\
			\qif $c_{\alpha}>0$ \qthen\label{L5}\\
				(* skip $S[i]$, it was ignored *)\\
				\qreturn $\DIST(i+1,j,\overline{c}-\overline{w}_{\alpha})$\label{L6}
			\qelif $\alpha=\beta$ \qthen\label{L7}\\
				(* $S[i]$ and $L[j]$ match *)\\
				\qreturn $\DIST(i+1,j+1,\overline{c})$\label{L8}
			\qelse\\
				$d_{ins}\leftarrow+\infty$, $d_{swaps}\leftarrow+\infty$\label{L9}\\
				\qif $c_{\beta}=0$ \qthen\\
					(* insert a $\beta$ at index $i$ *)\\
					$d_{ins}\leftarrow 1+\DIST(i,j+1,\overline{c})$\label{L-ins}
				\qfi\\			
				$r\leftarrow select(S,rank(S,i,\beta)+c_{\beta}+1,\beta)$\\
				\qif $r\neq null$ \qthen\\
					$\Delta\leftarrow \sum_{\theta=1}^{\D} \min\{c_{\theta},rank(S,r,\theta)-rank(S,i-1,\theta)\}$\\
					(* swaps *)\\
					$d_{swaps}\leftarrow (r-i)-\Delta+\DIST(i,j+1,\overline{c}+\overline{w}_{\beta})$\label{L-swaps}
				\qfi\\
				\qreturn $\min\{d_{ins},d_{swaps}\}$\label{L10} \label{line:min}
			\qfi
	\end{algorithm}
	\caption{{Informal algorithm to compute $\DIST(i,j,\overline{c})$:
		Lemma~\ref{lem:easy1} and
		Lemma~\ref{lem:easy2} guarantee the correctness of lines~\ref{L1} to~\ref{L4};
		Lemma~\ref{lem:core2} guarantees the correctness of lines~\ref{L5} to~\ref{L6};
		Lemma~\ref{lem:core1} guarantees the correctness of lines~\ref{L7} to~\ref{L8}; and
		Lemma~\ref{lem:core3} guarantees the correctness of lines~\ref{L9} to~\ref{L10}.
	}}
	\label{fig:code}
\end{figure}

\begin{LONG}
\subsection{Complete algorithm}
\label{sec:complete-algorithm}

In the following, we describe the formal algorithm to compute $\DIST(i,j,\overline{0})$. 
We consider the worst scenario for the running time of Theorem~\ref{theo:compute1},
where for each symbol $\alpha\in\Sigma$ we have $g_{\alpha}>0$. The other cases in which
$g_{\alpha}>0$ is not satisfied for all $\alpha\in\Sigma$ are easier to implement.
Note that the line~\ref{line:verifyDist} 
of algorithm \texttt{Compute} and line~\ref{line:verifyCounters} of algorithm 
$DIST2$ guarantee that $DIST2(i,j)=DIST(i,j,\overline{c})<+\infty$ in every call of $DIST2$.
Further, the counters $(c_1,c_2,\dots,c_d)$ are global variables to the recursive $DIST2$.  

\begin{figure}
\begin{algorithm}{}{{\tt Compute} $\dist(S,L)$:}
preprocess each of $S$ and $L$ for $rank$ and $select$\\
$(c_1,c_2,\dots,c_d)\leftarrow \overline{0}$\\
\qreturn $ ${\bf if} $\DIST(1,1,\overline{0})<+\infty$ {\bf then} $DIST2(1,1)$ {\bf else} $+\infty$\label{line:verifyDist} 
\end{algorithm}
	\caption{\small{Calling Algorithm to compute $\DIST(i,j,\overline{0})$, filtering degenerated cases before launching the real computation of the distance.}}
	\label{fig:code2}
\end{figure}

\begin{figure}
	\begin{algorithm}{}{$\DIST2(i,j)$:}
		$p\leftarrow $ the first index in $[1..\D]$ so that $c_p=0$\\
		\qfor $\alpha=1$ \qto $\D$ \qdo\\
			\qif $n_{\alpha}\le m_{\alpha}-n_{\alpha}$ \qthen\\
				$x_{\alpha} \leftarrow c_{\alpha}$
			\qelse\\
				$x_{\alpha} \leftarrow rank(L,j-1,\alpha) - rank(S,i-1,\alpha) - c_{\alpha}$
			\qfi
		\qrof\\
		$(r_1,\ldots,r_{\D-1})\leftarrow (x_1,\ldots,x_{p-1},x_{p+1},\ldots,x_{\D})$\\
		$k \leftarrow j-i-(r_1+\dots+r_{\D-1})$\\
		\qif $T[p,i,k,r_1,\ldots,r_{\D-1}]\neq undefined$ \qthen\\
			\qreturn $T[t,i,k,r_1,\ldots,r_{\D-1}]$
		\qelse\\
			\qif $i=n+1$ \qthen\\
				$T[p,i,k,r_1,\ldots,r_{\D-1}]\leftarrow m-j+1$
			\qelif $j=m+1$ \qthen\\
				$T[p,i,k,r_1,\ldots,r_{\D-1}]\leftarrow 0$
			\qelse\\
				$\alpha\leftarrow S[i]$, $\beta\leftarrow L[j]$\\
				\qif $c_{\alpha}>0$ \qthen\\
					$c_{\alpha}\leftarrow c_{\alpha}-1$\\
					$T[p,i,k,r_1,\ldots,r_{\D-1}]\leftarrow\DIST2(i+1,j)$\\
					$c_{\alpha}\leftarrow c_{\alpha}+1$
				\qelif $\alpha=\beta$ \qthen\\
					$T[p,i,k,r_1,\ldots,r_{\D-1}]\leftarrow\DIST2(i+1,j+1)$
				\qelse\\
					$d_{ins}\leftarrow+\infty$, $d_{swaps}\leftarrow+\infty$\\
					\qif $c_{\beta}=0$ \qand $count(S,i,\beta)< count(L,j,\beta)$ \qthen \label{line:verifyCounters}\\
						$d_{ins}\leftarrow 1+DIST2(i,j+1)$
					\qfi\\			
					$r\leftarrow select(S,rank(S,i,\beta)+c_{\beta}+1,\beta)$\\
					\qif $r\neq null$ \qthen\\
						$\Delta\leftarrow \sum_{\theta=1}^{\D} \min\{c_{\theta},rank(S,r,\theta)-rank(S,i-1,\theta)\}$\\
						$c_{\beta}\leftarrow c_{\beta}$+1\\
						$d_{swaps}\leftarrow (r-i)-\Delta+DIST2(i,j+1)$\\
						$c_{\beta}\leftarrow c_{\beta}-1$
					\qfi\\
					$T[p,i,k,r_1,\ldots,r_{\D-1}]\leftarrow\min\{d_{ins},d_{swaps}\}$
				\qfi
			\qfi\\
			\qreturn $T[p,i,k,r_1,\ldots,r_{\D-1}]$
		\qfi
	\end{algorithm}
	\caption{\small{Formal Algorithm to compute $\DIST(i,j,\overline{0})$, using dynamic 
	programming with memorization. Note that the line \ref{line:verifyCounters} 
	of algorithm \texttt{Compute} and line \ref{line:verifyDist} of algorithm 
	$DIST2$ guarantee that $DIST2(i,j)=DIST(i,j,\overline{c})<+\infty$ in every call.}}
	\label{fig:code3}
\end{figure}

\end{LONG}
 
\subsection{Complexity Analysis}
\label{sec:dp}

Combining Lemmas~\ref{lem:core1}~to~\ref{lem:easy2}, the value of
$\DIST(1,1,\overline{0})$ can be computed recursively, as shown in the
algorithm of Figure~\ref{fig:code}. 
We analyze the formal complexity of this algorithm in
Theorem~\ref{theo:compute1}, in the finest model that we can define,
taking into account the relation for each symbol $\alpha\in\Sigma$ between
the number $n_\alpha$ of occurrences of $\alpha$ in $S$ and the number
$m_\alpha$ of occurrences of $\alpha$ in $L$.

\begin{theorem}\label{theo:compute1}
Given two strings $S$ and $L$ over the alphabet $\Sigma$, for each symbol
$\alpha\in\Sigma$, note $n_\alpha$ the number of occurrences of $\alpha$
in $S$ and $m_{\alpha}$ the number of occurrences of $m$ in $L$, their
sums $n=n_1+\dots+n_{\D}$ and $m=m_1+\dots+m_{\D}$, and
$g_{\alpha}=\min\{n_{\alpha},m_{\alpha}-n_{\alpha}\}$ a measure of how far
$n_\alpha$ is from $m_\alpha/2$.
There is an algorithm computing the {\sc Swap-Insert Correction} distance
$\dist(S,L)$ in time within $O(d+m)$ if $S$ and $L$ have no symbol in
common, and otherwise in time within
\[
 O\left(
  d(n+m) 
  +d^2n
  \cdot  \K\cdot
    \prod_{\alpha\in\Sigma_{+}}(g_\alpha+1)
  \right),
\]
where $\Sigma_{+}=\{\alpha\in\Sigma: g_{\alpha}>0\}$ if $g_\alpha=0$ for
any $\alpha\in\Sigma$, and
$\Sigma_{+}=\Sigma\setminus \{\arg\min_{\alpha\in\Sigma}g_{\alpha}\}$
otherwise.
\end{theorem}

\begin{proof}
Observe first that there is a reordering of $\Sigma=[1..\D]$ such that $0<g_1\le
g_2\le\dots\le g_s$ and $g_{s+1}=g_{s+2}=\dots=g_{d}$ for some index $s\in[0..d]$,
and we assume such an ordering from now on.
%
Note also that given any string $X\in\{S,L\}$, a simple 2-dimensional
array using space within $O(\D\cdot |X|)$ can be computed in time within
$O(\D\cdot |X|)$, to support the queries $rank(X,i,\alpha)$ and
$select(X,k,\alpha)$  in constant time for all values of $i\in[1..n]$,
$k\in[1..|X|]$, and $\alpha\in\Sigma$.

The case where the two strings $S$ and $L$ have no symbol in common is
easy: the distance is then $+\infty$.  The algorithm detects this case by
testing if $g_\alpha=0$ for all $\alpha\in\Sigma$, in time within
$O(d+m)$.

Consider the algorithm of Figure~\ref{fig:code}, and let $i\in[1..n]$,
$j\in[1..m]$, and $\overline{c}=(c_1,\ldots,c_{\D})$ be parameters such
that $\DIST(i,j,\overline{c})<+\infty$.

At least one of the $c_1,\ldots,c_{\D}$ is equal to zero: in the first
entry $\DIST(1,1,\overline{0})$ all the counters $c_1,c_2,\ldots,c_{\D}$
are equal to zero, and any counter is incremented only at
line~\ref{L-swaps}, in which another counter must be equal to zero because
of the lines~\ref{L5} and~\ref{L7}.

The number of \texttt{insertions} counted in line~\ref{L-ins}, in previous
calls to the function $\DIST$ in the recursion path from
$\DIST(1,1,\overline{0})$ to $\DIST(i,j,\overline{c})$, is equal to
$j-i-(c_1+\dots+c_{\D})$.
Let $t_{\alpha}$ denote the number of such insertions for the symbol
$\alpha\in\Sigma$. Then, we have 
\[
	j~=~i+(c_1+\dots+c_{\D})+(t_1+\dots+t_{\D}),
\]
and for all $\alpha\in\Sigma$, $c_{\alpha}\leq n_{\alpha}$, $t_{\alpha}\leq
m_{\alpha}-n_{\alpha}$, and
\[
	c_{\alpha}+t_{\alpha} ~=~ rank(L,j-1,\alpha)-rank(S,i-1,\alpha).
\]
Using the above observations, we encode all entries
$\DIST(i,j,\overline{c})$, for $i,j$ and $\overline{c}$ such
that $\DIST(i,j,\overline{c})<+\infty$, into the following table $T$ of
$s+2\le \D+2$ dimensions.
If we have $s=d$, then
\[
	T[p,i,k,r_1,\ldots,r_{\D-1}]~=~\DIST(i,j,\overline{c}=(c_1,\ldots,c_{\D})),
\]
where
\begin{align*}
	c_p & ~=~ 0,\\
	(r_1,\ldots,r_{\D-1}) & ~=~ (x_1,\ldots,x_{p-1},x_{p+1},\ldots,x_{\D})\\
	x_{\alpha} & ~=~ \left\{ 
		\begin{array}{ll}
			c_{\alpha} & \text{if }~ n_{\alpha}\le m_{\alpha}-n_{\alpha}\\
			t_{\alpha} & \text{if }~ m_{\alpha}-n_{\alpha} < n_{\alpha}
		\end{array}
	\right.~\text{for every } \alpha\in\Sigma, \text{ and} \\
	k & ~=~ (c_1+\dots+c_{\D})+(t_1+\dots+t_{\D})-(r_1+\dots+r_{\D-1}).
\end{align*}
Furthermore, given any combination of values $i,j,c_1,\ldots,c_{\D}$ we can switch
to the values $p,i,k,r_1,\ldots,r_{\D-1}$, and vice versa, in time within $O(d)$.
Otherwise, if $s<d$, then
\[
	T[i,k,r_1,\ldots,r_{s}]~=~\DIST(i,j,\overline{c}=(c_1,\ldots,c_{\D})),
\]
where $(r_1,\ldots,r_{s}) = (x_1,\ldots,x_{s})$. Again, 
given the values $i,j,c_1,\ldots,c_{\D}$ we can switch
to the values $i,k,r_1,\ldots,r_{s}$, and vice versa, in $O(d)$ time.

Since $p\in [1..\D]$, $i\in[1..n+1]$, $k\in [0..\K]$, 
and $r_{\alpha}\in
[0..g_{\alpha}]$ for every $\alpha$, the table $T$ can be as large as
$\D\times (n+1)\times (1+\K)\times (g_2+1)\times\dots\times(g_{\D}+1)$ if $s=d$, and
as large as $(n+1)\times  (1+\K)\times (g_1+1)\times\dots\times(g_{s}+1)$ if $0<s<d$. 
For $s=0$, no table is needed.
\begin{INUTILE}
Using a modified version of the algorithm of Figure~\ref{fig:code}, based
on memoization on the table $T$, the goal entry $\DIST(1,1,\overline{0})$
can be computed recursively, where each entry $\DIST(i,j,\overline{c})$ is
computed in (amortized) $O(\D)$ time.
\end{INUTILE}
The running time of this new algorithm includes the $O(d(n+m))=O(\D m)$ time for
processing each of $S$ and $L$ for $rank$ and $select$, and the time to
compute $DIST(1,1,\overline{0})$ which is within $O(\D)$ times $n+m$ plus the number
of cells of the table $T$. If $s=d$, the time to compute $DIST(1,1,\overline{0})$ is
within
\[
	O\left(d(n+m)+{\D}^2 n \cdot \K \cdot(g_2+1)\cdot\dots\cdot(g_{\D}+1)\right).
\]
Otherwise, if $0\le s<d$, the time to compute $DIST(1,1,\overline{0})$ is within
\[
	O\left(d(n+m)+{\D} n \cdot \K \cdot(g_1+1)\cdot\dots\cdot(g_{s}+1)\right).
\]
The result follows by noting that: if $s=d$, then $\Sigma_+=\{2,\ldots,d\}$. Otherwise,
if $s<d$, then $\Sigma_+=\{1,\ldots,s\}$.
\qed\end{proof}

The result above, about the complexity in the worst case over instances
with $d,n_1,\ldots,n_d$, $m_1,\ldots,m_d$ fixed, implies results in less
precise models, such as in the worst case over instances for $d,n,m$
fixed:

\begin{corollary}\label{theo:compute2}
Given two strings $S$ and $L$ over the alphabet $\Sigma$, of respective
sizes $n$ and $m$, the algorithm analyzed in Theorem~\ref{theo:compute1}
computes the {\sc Swap-Insert Correction} distance $\dist(S,L)$ in time
within
\[
	O\left(\D(n+m)+ {\D}^2n(m-n)\left(\frac{n}{\D-1}+1\right)^{\D-1}\right),
\]
which is within $O\left(n+m+ n^d(m-n)\right)$ for alphabets of fixed size $d$; and within
\[
	O\left(\D(n+m)+ {\D}^2n^2\left(\frac{m-n}{\D-1}+1\right)^{\D-1}\right),
\]
which is within $O\left(n+m+ n^2(m-n)^{d-1}\right)$ for alphabets of fixed
size $d$.
\end{corollary}

\begin{LONG}
\begin{proof}
We use the following claim: If $a\ge 1$ and $x\le y$, then $(a+y)(x+1)\le (a+x)(y+1)$.
It can be proved as follows:
\begin{eqnarray*}
	(a-1)x & \le & (a-1)y\\
	ax + y & \le & ay + x\\
	ax + y + a + xy & \le & ay + x + a + xy\\
	(a+y)(x+1) & \le & (a+x)(y+1).
\end{eqnarray*}
Let $\Sigma_{+}\subset \Sigma$ be as defined in Theorem~\ref{theo:compute1},
and consider the worst scenario for the running time, that is, let us consider
w.l.o.g.\ that $\Sigma_{+}=[2..d]$.
Let $\beta\in\Sigma_{+}$ be a symbol such that $m_{\beta}-n_{\beta}<n_{\beta}$,
and define $a$ and $b$ such that:
\begin{eqnarray*}
   a  & = & ~\K-(m_{\beta}- g_{\beta}) ~=~ \K-n_{\beta}\\
      & = & \sum_{\alpha\in\Sigma\setminus\{\beta\}}(m_{\alpha}-g_{\alpha}) ~\ge~ 1,
\end{eqnarray*}
and
\[
	b ~=~ \prod_{\alpha\in\Sigma_{+}\setminus\{\beta\}}(g_{\alpha}+1). 
\]
Note that:
\begin{eqnarray*}
	\K\cdot\prod_{\alpha\in\Sigma_{+}}(g_{\alpha}+1) & = & (a + n_{\beta})\cdot b \cdot (m_{\beta}-n_{\beta} + 1)\\
		& \le & (a + m_{\beta}-n_{\beta})\cdot b \cdot (n_{\beta} + 1),
\end{eqnarray*}
which immediately implies
\[
	\K\cdot\prod_{\alpha\in\Sigma_{+}}(g_{\alpha}+1) ~\le~ (m-n)\prod_{\alpha\in\Sigma_{+}}(n_{\alpha}+1).
\]
Then, 
\begin{eqnarray*}
	O\left({\D}^2n\cdot \K\cdot\prod_{\alpha\in\Sigma_{+}}(g_{\alpha}+1)\right) 
	& \subseteq & O\left({\D}^2n(m-n)\cdot(n_2+1)\cdot\dots\cdot(n_{\D}+1)\right)\\
	& \subseteq & O\left({\D}^2 n(m-n)\cdot \left(\frac{n_2+\dots+n_{\D}}{\D-1}+1\right)^{\D-1}\right)\\
	& \subseteq & O\left({\D}^2 n(m-n)\cdot \left(\frac{n}{\D-1}+1\right)^{\D-1}\right).
\end{eqnarray*}
Using similar arguments, we can prove that
\[
	\K\cdot\prod_{\alpha\in\Sigma_{+}}(g_{\alpha}+1) ~\le~ n\prod_{\alpha\in\Sigma_{+}}(m_{\alpha}-n_{\alpha}+1),
\]
which implies the second part of the result.
\qed\end{proof}
\end{LONG}

\section{Discussion}\label{sec:conclusions}

\begin{LONG}
In 2014, Meister~\cite{2015-TCS-UsingSwapsAndDeletesToMakeStringsMatch-Meister} described an algorithm computing
the \textsc{Swap-Insert Correction} distance from a string
$S\in[1..\alphabetSize]^n$ to another string
$L\in[1..\alphabetSize]^m$ on any fixed alphabet size
$\alphabetSize\ge 2$, in time polynomial in $n$ and $m$.
\end{LONG}
\begin{LONG}
The algorithm that we described takes advantage of instances where for all
symbols $\alpha\in\Sigma$ the number $n_\alpha$ of occurrences of $\alpha$
in $S$ is either close to zero (i.e. most $\alpha$ symbols from $L$ are
placed in $S$ through \texttt{insertions}) or close to the number
$m_\alpha$ of occurrences of $\alpha$ in $L$ (i.e. most $\alpha$ symbols
from $L$ are matched to symbols in $S$ through \texttt{swaps}), while
still running in time within
$O(m+ \min\left\{n^d(m-n),n^2(m-n)^{d-1}\right\})$ when the alphabet size
$\D$ is a constant, in the worst case over instances composed of strings
of sizes $n$ and $m$.
\end{LONG}
The exact running time of our algorithm is within
\[
 O\left(
  d(n+m) 
  +d^2n
  \cdot  \K\cdot
    \prod_{\alpha\in\Sigma_{+}}(g_\alpha+1)
  \right),
\]
where 
$n_\alpha$ and $m_\alpha$ are the respective number of
occurrences of symbol $\alpha\in[1..d]$ in $S$ and $L$ respectively;
where the vector formed by the values
$g_\alpha=\min\{n_\alpha,m_\alpha-n_\alpha\}$ measures the distance
between $(n_1,\ldots,n_\sigma)$ and $(m_1,\ldots,m_\sigma)$; and 
where $\Sigma_{+}=\{\alpha\in\Sigma: g_{\alpha}>0\}$ if $g_\alpha=0$ for
any $\alpha\in\Sigma$, and
$\Sigma_{+}=\Sigma\setminus \{\arg\min_{\alpha\in\Sigma}g_{\alpha}\}$
otherwise.

Summarizing the disequilibrium between the frequency distributions of the
symbols in the two strings via the measure
$g=\max_{\alpha\in\Sigma} g_\alpha\leq n$, this yields a complexity within
$O(d^2 n m g^{d-1})$, which is polynomial in $n$ and $m$, and exponential
only in $d$ of base $g$.
Since this disequilibrium $g$ is smaller than the length $n$ of the
smallest string $S$, this implies a worst case complexity within
$O(d^2 m n^d)$ over instances formed by strings of lengths $n$ and $m$
over an alphabet of size $d$, a result matching the state of the
art~\cite{2015-TCS-UsingSwapsAndDeletesToMakeStringsMatch-Meister} for this problem.

\subsection{Implicit Results}

The result from Theorem~\ref{theo:compute1} implies the following
additional results:

\begin{description}

\item[Weighted Operators:] Wagner and Fisher \cite{wagner1974string}
  considered variants where the cost $c_{ins}$ of an \texttt{insertion}
  and the cost $c_{swap}$ of an \texttt{swap} are distinct. In the
  \textsc{Swap-Insert Correction} problem, there are
  always $n-m$ \texttt{insertions}, and always $\dist(S,L)-n+m$
  \texttt{swaps}, which implies the optimality of the algorithm we
  described in such variants.
 
  \begin{LONG}
\item[Computing the Sequence of Corrections:] Since any correct
  algorithm must verify the correctness of its output, given a set
  $\cal C$ of correction operators, any correct algorithm computing
  the \textsc{String-to-String Correction Distance} when limited to
  the operators in $\cal C$ implies an algorithm computing a minimal
  sequence of corrections under the same constraints within the same
  asymptotic running time.
  \end{LONG}

\item[Implied improvements when only swaps are needed:] Abu-Khzam et
  al.~\cite{AbuKhzam201141} mention an algorithm computing the
  \textsc{Swap String-to-String Correction} distance (i.e. only
  \texttt{swaps} are allowed) in time within $O(n^2)$.  This is a
  particular case of the \textsc{Swap-Insert Correction}
  distance, which happens exactly when the two strings are of the same
  size $n=m$ (and no \texttt{insertion} is neither required nor allowed).
  In this particular case, our algorithm yields a solution running in time
  within $O(\D m)$, hence improving on Abu-Khzam et al.'s
  solution~\cite{AbuKhzam201141}.

\item[Effective Alphabet:] Let $\alphabetSize'$ be the effective
  alphabet of the instance, i.e. the number of symbols $\alpha$ of
  $\Sigma=[1..\alphabetSize]$ such that the number of occurrences of
  $\alpha$ in $S$ is a constant fraction of the number of occurrences
  of $\alpha$ in $L$ (i.e. $n_\alpha\in \Theta(m_\alpha)$).  Our
  result implies that the real difficulty is $\alphabetSize'$ rather
  than $\alphabetSize$, i.e. that even for a large alphabet size
  $\alphabetSize$ the distance can still be computed in reasonable
  time if $\alphabetSize'$ is finite.

\end{description}

\subsection{Perspectives}

Those results suggest various directions for future research:
\begin{description}

\item[Further improvements of the algorithm:] our algorithm can be
  improved further using a lazy evaluation of the $\min$ operator on
  line \ref{line:min}, so that the computation in the second branch of
  the execution stops any time the computed distance becomes larger
  than the distance computed in the first branch. This would save time
  in practice, but it would not improve the worst-case complexity in
  our analysis, in which both branches are fully explored: one would
  require a finer measure of difficulty to express how such
  a modification could improve the complexity of the algorithm

\item[Further improvements of the analysis:] The complexity of
  Abu-Khzam et al.'s algorithm~\cite{AbuKhzam201141}, sensitive to the
  distance from $S$ to $L$, is an orthogonal result to ours. An
  algorithm simulating both their algorithm and ours in parallel
  yields a solution adaptive to both measures, but an algorithm using
  both techniques in synergy would outperform both on some instances,
  while never performing worse on other instances.

\item[Adaptivity for other existing distances:] Can other
  \textsc{String-to-String Correction} distances be computed faster
  when the number of occurrences of symbols in both strings are
  similar for most symbols? Edit distances such as when only
  insertions or only deletions are allowed are linear anyway, but more
  complex combinations require further studies.

\end{description}


%% file: 2015-ARXIV-AdaptiveComputationOfTheSwapInsertEditionDistance-BarbayPerez.bbl
\begin{thebibliography}{1}

\bibitem{AbuKhzam201141}
F.~N. Abu-Khzam, H.~Fernau, M.~A. Langston, S.~Lee-Cultura, and U.~Stege.
\newblock Charge and reduce: A fixed-parameter algorithm for {String-to-String
  Correction}.
\newblock {\em Discrete Optimization (DO)}, 8(1):41 -- 49, 2011.

\bibitem{2015-SPIRE-AdaptiveComputationOfTheSwapInsertEdutionDistance-BarbayPerez}
J.~Barbay and P.~P{\'e}rez-Lantero.
\newblock Adaptive computation of the {Swap-Insert Edition Distance}.
\newblock In {\em Proceedings of the Symposium on String Processing and
  Information Retrieval (SPIRE)}, 2015.
\newblock (to appear).

\bibitem{Cormen2009}
T.~H. Cormen, C.~E. Leiserson, R.~L. Rivest, and C.~Stein.
\newblock {\em Introduction to Algorithms, Third Edition}.
\newblock The MIT Press, 3rd edition, 2009.

\bibitem{2015-TCS-UsingSwapsAndDeletesToMakeStringsMatch-Meister}
Daniel Meister.
\newblock Using swaps and deletes to make strings match.
\newblock {\em Theoretical Computer Science (TCS)}, 562(0):606 -- 620, 2015.

\bibitem{spreen2013}
T.~D. Spreen.
\newblock The {Binary String-to-String Correction Problem}.
\newblock Master's thesis, University of Victoria, Canada, 2013.

\bibitem{wagner1975complexity}
R.~A. Wagner.
\newblock On the complexity of the extended {String-to-String Correction
  Problem}.
\newblock In {\em Proceedings of the seventh annual ACM {Symposium on Theory Of
  Computing} ({STOC})}, pages 218--223. ACM, 1975.

\bibitem{wagner1974string}
R.~A. Wagner and M.~J. Fischer.
\newblock The {String-to-String Correction Problem}.
\newblock {\em Journal of the ACM (JACM)}, 21(1):168--173, 1974.

\bibitem{wagner1975extension}
R.~A. Wagner and R.~Lowrance.
\newblock An extension of the {String-to-String Correction Problem}.
\newblock {\em Journal of the ACM (JACM)}, 22(2):177--183, 1975.

\bibitem{watt2013}
N.~Watt.
\newblock {String to String Correction} kernelization.
\newblock Master's thesis, University of Victoria, Canada, 2013.

\end{thebibliography}
